\documentclass[dvipdfm,12pt]{article}
\usepackage[dvips]{graphicx}
\usepackage[usenames,dvipsnames]{color}
\usepackage[dvipdfm]{hyperref}
\usepackage{amsmath,amsthm,amssymb,euscript,amscd}

\numberwithin{equation}{section} 

\title{\bf h-analogue of Fibonacci Numbers} 
\author{H.B. Benaoum\footnote{Email : hbenaoum@pmu.edu.sa, hbenaoum@physics.syr.edu} \\
{Prince Mohammad University, Al-Khobar 31952, Saudi Arabia}
}
\date{ }

\begin{document} 

\baselineskip24pt
\maketitle 

\abstract{In this paper, we introduce the h-analogue of Fibonacci numbers for non-commutative h-plane. 
For $h h'= 1$ and $h = 0$, these are just the usual Fibonacci numbers as it should be. We also derive a collection of identities for these numbers. Furthermore, h-Binet's formula for the h-Fibonacci numbers is found and the generating function that generates these numbers is obtained.}

~
\newtheorem{theorem}{Theorem}
\newtheorem{corollary}[theorem]{Corollary}
\newtheorem{lemma}[theorem]{Lemma}
\newtheorem{proposition}[theorem]{Proposition}
\newtheorem{conjecture}[theorem]{Conjecture}
\newtheorem{rem}[theorem]{Remark}

\newtheoremstyle{defi}
  {10pt}		  
  {10pt}  
  {\rm}  
  {\parindent}     
  {\bf}  
  {. }    
  { }    
  {}     
\theoremstyle{defi}
\newtheorem{definition}[theorem]{Definition}

\def\proofname{\indent {\sl Proof.}}

~\\
~\\ 
~\\
~\\
~\\
\begin{center}
{\em \Large To my wife Nawal} 
\end{center}

\newpage 
\section{ Introduction }  
Fibonacci recursive sequence has fascinated scholars and amateurs for centuries. Since their appearance in the book {\em Liber Abaci } published in 1202 by the Italian medieval mathematician Leonardo Fibonacci, they have been encountered in many distincts contexts, ranging from the arts, pure mathematics and physical sciences to electrical engineering. \\

These numbers are not random numbers but each number is made by adding the previous number to the present one. The ratio of successive pairs tends to the so-called golden ratio $\varphi = 1.618033989 \dots$ and whose reciprocal is $0.618033989 \dots$ , so that we have 
$\varphi = 1 + \varphi^{-1}$. This golden ratio is an irrational number with several curious properties and one can come across this ratio in many areas of arts and sciences. In fact, the ancient Greeks found it a very interesting and divine number. Binet's formula for Fibonacci numbers is exceptional because it is expressed in terms of irrational number, even though all Fibonacci numbers are integers. \\

Fibonacci numbers have been studied both for their applications and the mathematical beauty of the rich and varied identities that they satisfy. 
The reference ~\cite{koshy} contains many results on Fibonacci numbers with detailed proofs. Interestingly enough, the amazing Fibonacci numbers seem to be intrinsic in nature since they have been identified in ( leaves, petal arrangements, pinecones, pineapples, seeds and shells).  \\

The Fibonacci sequence is defined by the recurrence : 
\begin{eqnarray}
f_{n+1} & = & f_n + f_{n-1} 
\end{eqnarray}
with the initial conditions $f_0 = 0$ and $f_1 = 1$. \\

It is easy to deduce the following identity that connects Fibonacci numbers $f_n$ and the binomial coefficients : 
\begin{eqnarray}
f_n & = & \sum_{k=0}^{[\frac{n-1}{2}]}  \left( \begin{array}{c} 
n - 1 -  k \\
k  \\  \end{array} 
\right),~~~~~n>0 
\end{eqnarray}   
Similarly, the Fibonacci polynomials are defined as : 
\begin{eqnarray}
f_n \left( u , v \right) & = & \sum_{k=0}^{[\frac{n}{2}]} \left( \begin{array}{c} 
n -1 -k \\
k \\ \end{array} 
\right) u^{n - 2 k - 1} v^k 
\end{eqnarray}

The main objective of this paper is to generalize the results on classical Fibonacci numbers to non-commutative h-plane. That is to introduce the h-analogue of Fibonacci numbers on the non-commutative h-plane using the h-binomial coefficients \cite{benaoum}. In section 2, we define the h-Pascal triangle which arises naturally from the h-binomial coefficients and prove some elegant identities which we relate to the Charlier polynomials. Section 3 focuses on the h-Fibonacci numbers by providing the necessary definition and shows that the known Fibonacci identities and their bijective proofs can be easily extended to bijective proofs of h-analogues of these identities. Section 4 is devoted to reformulate the h-Fibonacci sequence in terms of a matrix representation. h-Binet's formula for the h-Fibonacci operators is introduced in section 5 and a number of identities are derived. Finally, the generating function of the h-Fibonacci sequence is obtained in the last section. 
This section also provides us with lists the generating functions for the various powers and products of h-Fibonacci sequence. \\

In \cite{benaoum} , we introduced the h-analogue of Newton's binomial formula :
\begin{eqnarray}
\left( x + y \right)^n & = & \sum_{k=0}^n \left[ 
\begin{array}{l} 
n \\
k   
\end{array} \right]_{h,h'} y^k x^{n-k} 
\end{eqnarray} 
where x and y are non-commuting variables satisfying :  
\begin{eqnarray}
x y & = & y x + h y^2 
\end{eqnarray}  

Here $\left[ \begin{array}{l} 
n \\
k \end{array} \right]_{h,h'}$ is the h-binomial coefficients given as follows : 
\begin{eqnarray}
\left[ 
\begin{array}{l} 
n \\
k \end{array} \right]_{h,h'} & = & \left( \begin{array}{l} 
n \\
k \end{array} \right)~ h^k \left(  h' \right)_{1;k} \nonumber \\
& = & \left( \begin{array}{l} 
n \\
k \end{array} \right)~\left(  h h' \right)_{h;k}
\end{eqnarray} 
with $h h' = 1$. \\

$(a)_k$ is the shifted factorial defined as : 
\begin{eqnarray}
\left(a \right)_{s;n} & = &  \left\{ \begin{array}{ll} 
1 &  n = 0 \\
a \left( a + s \right) \ldots \left( a + (n - 1) s \right) & n = 1,2, \ldots \\
\end{array} 
\right. 
\end{eqnarray} 

In what follows, we consider the h-binomial coefficients with two parameters $h$ and $h'$ such that $h h'$ is not necessarily equal to $1$. 
These coefficients obey to the following properties :  

\begin{eqnarray}
\left[ 
\begin{array}{l} 
n + 1 \\
k \end{array} \right]_{h,h'} & = &  
\left[ 
\begin{array}{l} 
n  \\
k \end{array} \right]_{h,h'} + h h' \left[ 
\begin{array}{l} 
n  \\
k - 1  \end{array} \right]_{h,h'+ 1} 
\label{equation18}
\end{eqnarray}
and 
\begin{eqnarray}
\left[ 
\begin{array}{l} 
n + 1 \\
k + 1 \end{array} \right]_{h,h'} & = &  h h'~\frac{n+1}{k+1}~
\left[ 
\begin{array}{l} 
n  \\
k \end{array} \right]_{h,h'+1} 
\end{eqnarray}

\section{h-Pascal triangle} 
The h-Pascal triangle is constructed by considering the h-binomial coefficient of the nth 
column $\left( n = 0, 1, 2, 3, \cdots \right)$ and the kth row $\left( k = 0, 1, 2, 3, \cdots \right)$.  \\

\begin{table}[!hbp]
\begin{center}
\begin{tabular}{c|ccccccccc} 
$n \setminus k$ & 0 & 1 & 2  & 3  & 4 &  5  &  6 & 7  & \dots  \\ 
\hline 
0     & 1 &   &    &    &    &    &    &    &        \\ 
1     & 1 & 1 &    &    &    &    &    &    &        \\
2     & 1 & 2 & $(h h')_{h;1}$  &    &    &    &    &    &        \\ 
3     & 1 & 3 & 3 $(h h')_{h;1}$ & $( h h')_{h;2}$  &    &    &    &    &        \\
4     & 1 & 4 & 6 ($h h')_{h;1}$ & 4 $(h h')_{h;2}$  &  $(h h')_{h;3}$ &    &    &    &        \\
5     & 1 & 5 & 10 $( h h')_{h;1}$  & 10 $(h h')_{h;2}$  & 5 $(h h')_{h;3}$  & $(h h')_{h;4}$  &    &    &      \\ 
6     & 1 & 6 & 15 $(h h')_{h;1}$ & 20 $(h h')_{h;2}$ & 15 $(h h')_{h;3}$ & 6 $(h h')_{h;4}$  &  $(h h')_{h;5}$ &    &     \\ 
7     & 1 & 7 & 21 $(h h')_{h;1}$ & 35 $(h h')_{h;2}$ & 35 $(h h')_{h;3}$ & 21 $(h h')_{h;4}$ &  7 $(h h')_{h;5}$ & $(h h')_{h;6}$  &     \\ 
\end{tabular} 
\caption{h-Pascal triangle}  
\label{table2}
\end{center}  
\end{table}  

If we sum the h-binomial coefficients of the h-Pascal triangle, the following identity that connects the h-binomial coefficients 
to Charlier polynomials is obtained : 
\begin{eqnarray}
\sum_{k=0}^n \left[ \begin{array}{l} 
n \\ 
k \end{array} \right]_h & = & c_n ( -h', h^{-1} ) 
\end{eqnarray} 
where $c_n ( z, a )$ are the Charlier polynomials defined by the following formula \cite{chihara} : 
\begin{eqnarray}
c_n ( z, a ) & = & \sum_{k=0}^n \left( \begin{array}{l} 
n \\ 
k \end{array} \right) a^{-k} (- z )_{1;k}
\end{eqnarray}
Another identity satisfied by the h-Pascal Triangle is the sum of all elements in a column is given by : 
\begin{eqnarray}
h (h' +j ) ~\sum_{i=1}^n \left[ \begin{array}{l} 
i \\ 
j \end{array} \right]_{h,h'} & = & \left[ \begin{array}{l} 
n + 1 \\ 
j + 1 \end{array} \right]_{h,h'} 
\end{eqnarray}
 
\section{h-Fibonacci numbers } 
In this section, the h-Fibonacci numbers are introduced. It should be noted that the recurrence formula of these numbers depends on the parameters 
$h$ and $h'$. They reduced to the usual Fibonacci numbers when $h h' = 1$ and $h$ goes to zero. \\
  
The h-Fibonacci numbers which are obtained by adding diagonal numbers of the h-Pascal triangle, are given by : 
\begin{eqnarray}
F_n^{(h,h')} & = & \sum_{k=0}^{[\frac{n-1}{2}]} \left[ \begin{array}{c} 
n - 1 - k \\ 
k  \\ \end{array} \right]_{h,h'},~~~~~n>0 
\end{eqnarray}   

Since hypergeometric functions are important tool in many branches of pure and applied mathematics, a direct connection between h-Fibonacci numbers and hypergeometric functions can be established. Indeed we have : 
\begin{eqnarray} 
F_n^{(h,h')} & = & {}_3F_1(-n/2+1/2,-n/2+1,h';-n+1;-4h)
\end{eqnarray} 


The list the first 10 h-Fibonacci, when expanded in powers series of $h$ and $h'$, are shown in the table below : \\ 

\begin{table}[!hpb]
\begin{center}
\begin{tabular}{|c|c|c|}
\hline 
$n$ & $F_n^{(h,h')}$ & $f_n$ \\ 
\hline
0 & 0 & 0 \\ 
\hline
1 & 1 & 1  \\ 
\hline
2 & 1 & 1 \\  
\hline 
3 & $1 + h h'$ & 2 \\ 
\hline
4 & $1 + 2 h h'$ & 3  \\ 
\hline 
5 & $1 + 3 h h' + h^2 h' (h' + 1)$ & 5  \\ 
\hline 
6 & $1 + 4 h h' + 3 h^2 h'(h'+1)$ & 8 \\  
\hline 
7 & $1 + 5 h h' + 6 h^2 h' (h' + 1) + h^3 h' (h' + 1) (h' + 2 )$ & 13 \\ 
\hline 
8 & $1 + 6 h h' + 10 h^2 h' (h'+1) + 4 h^3 h' (h'+1 ) (h' + 2 ) $ & 21 \\ 
\hline  
9 & $1 + 7 h h' + 15 h^2 h' (h' + 1 ) + 10 h^3 h' (h' + 1 ) (h' + 2) + h^3 h' (h' + 1) ( h' + 2  ) (h' + 3)$ & 34 \\ 
\hline
10 & $1 + 8 h h'+ 21 h^2 h' (h' + 1 ) + 20 h^3 h' (h' + 1 ) (h' + 2) + 5 h^4 h' (h' + 1 ) (h' + 2 ) (h' + 3)$ & 55 \\  
\hline
\end{tabular} 
\end{center} 
\caption{First 10 numbers of h-Fibonacci and Fibonacci numbers.}
\label{table3} 
\end{table} 

We also provided in the table a list of the classical Fibonacci numbers just to compare both of them. As it was anticipated, the $F_n^{( h,h' )}$ sequence reduces to the usual $f_n$ when $h h' = 1$ and $h$ goes to zero.

\begin{theorem} ( h-Fibonacci recurrence ) 
~\\
The h-Fibonacci numbers obey the following recurrence formula 
\begin{eqnarray}
F_{n+1}^{(h,h')} & = & F_n^{(h,h')} + h h' ~F_{n-1}^{( h, h' + 1)} 
\end{eqnarray} 
\end{theorem} 

\begin{proof}  
~\\
Using equation (\ref{equation18}), we have :  
\begin{eqnarray*}
F_{n+1}^{( h, h')} & = & \sum_{k=0}^{[\frac{n}{2}]} \left[ \begin{array}{c} 
n - k \\ 
k  \\ \end{array} \right]_{h,h'} =  \sum_{k=0}^{[\frac{n-1}{2}]} \left[ \begin{array}{c} 
n - 1 - k \\ 
k  \\ \end{array} \right]_{h,h'} + h h' ~\sum_{k=1}^{[\frac{n-1}{2}]} \left[ \begin{array}{c} 
n - 1 - k \\ 
k - 1  \\ \end{array} \right]_{h,h'+1} \\  
& = & F_n^{( h,h' )} + h h' ~\sum_{k=0}^{[\frac{n-2}{2}]} \left[ \begin{array}{c} 
n - 2 - k \\ 
k   \\ \end{array} \right]_{h,h'+1} = F_n^{( h, h')} + h h' ~F_{n-1}^{(h, h'+1)} 
\end{eqnarray*}  
\end{proof} 

As the usual Fibonacci numbers, the h-Fibonacci numbers satisfy  numerous identities. We express some of them below :  \\

\begin{theorem}  
~\\
The h-Fibonacci numbers have the property :
\begin{eqnarray}
h h' ~\sum_{k=1}^n F_k^{(h, h' +1 )} & = & F_{n+2}^{( h, h' +1 )} - 1 
\end{eqnarray} 
\end{theorem} 

\begin{proof} 
~\\
Using the h-Fibonacci recurrence relation, we have : 
\begin{eqnarray*}
h h' ~F_n^{(h, h'+1 )} & = & F_{n+2}^{( h, h' )} - F_{n+1}^{( h, h' )} \\ 
h h' ~F_{n-1}^{(h, h'+1 )} & = & F_{n+1}^{( h, h+1 )} - F_n^{( h, h' )} \\ 
h h' ~F_{n-2}^{( h, h' + 1 )} & = & F_n^{( h, h' )} - F_{n-1}^{( h, h' )} \\ 
& \vdots &    \\ 
h h' ~F_3^{(h, h'+1)} & = & F_5^{( h, h' )} - F_4^{( h, h' )} \\
h h' ~F_2^{(h, h'+1 )} & = & F_4^{( h, h'+1 )} - F_3^{( h, h'+1 )} \\
h h'~F_1^{(h, h'+1 )} & = & F_3^{( h, h' )} - F_2^{( h, h' )} 
\end{eqnarray*}  
 
Adding all these equations, we get : 
\begin{eqnarray*}
h h' ~\sum_{k=1}^n F_k^{(h, h'+1)} & = & F_{n+2}^{( h, h'+1 )} - F_2^{( h, h'+1)} \\
& = & F_{n+2}^{( h, h' )} - 1 
\end{eqnarray*} 
\end{proof} 

\begin{theorem} 
~\\
The following identity holds 
\begin{eqnarray} 
\sum_{k=1}^n h^{n-k} (h')_{1;{n-k}} ~F_{2 k - 1}^{(h, h' + n - k )} & = & F_{2n}^{(h,h')} 
\end{eqnarray} 
\end{theorem} 

\begin{proof} 
~\\ 
Using the h-Fibonacci recurrence relations, we have : 
\begin{eqnarray*} 
F_{2 n - 1}^{(h,h')} & = & F_{2 n}^{(h,h')} - h h' ~F_{2 n - 2}^{(h,h' + 1)} \\  
F_{2 n - 3}^{(h,h'+1)} & = & F_{2 n - 2}^{(h, h' + 1)} - h (h'+1) ~F_{2 n - 4}^{(h, h' +2 )} \\ 
F_{2 n - 5}^{(h,h'+2)} & = & F_{2 n - 4}^{(h, h' + 2)} - h (h'+2) ~F_{2 n - 6}^{(h, h' +3 )} \\ 
& \vdots &  \\  
F_{2 k - 1}^{(h,h'+ n -k)} & = & F_{2 k}^{(h, h' + n -k)} - h (h'+n -k) ~F_{2 k - 2}^{(h, h' + n - k + 1 )} \\
& \vdots &  \\
F_5^{(h,h' + n - 3)} & = & F_6^{(h, h' + n - 3 )} - h (h' + n - 3 ) ~F_4^{(h, h' + n - 2)} \\
F_3^{(h,h' + n - 2)} & = & F_4^{(h, h' + n - 2 )} - h (h' + n - 2 ) ~F_2^{(h, h' + n - 1)} \\
F_1^{(h,h' + n - 1)} & = & F_2^{(h, h' + n - 1 )} - h (h' + n - 1 ) ~F_0^{(h, h' + n) )} \\
\end{eqnarray*}
Multiplying $F_{2 k - 1}^{(h,h'+ n -k)}$ by $h^{n-k} (h')_{1;{n-k}}$ and adding these equations, we get : 
\begin{eqnarray*}
\sum_{k=1}^n h^{n-k} (h')_{1;{n-k}} ~F_{2 k - 1}^{(h,h'+ n - k)} & = & F_{2 n}^{(h,h')} - h^{n} (h')_{1;{n}} ~F_0^{(h, h' + n + 1)} \\
& = & F_{2 n}^{(h, h')} 
\end{eqnarray*} 
\end{proof}
\begin{theorem} 
~\\
For $n > 0$, we have the following property of the h-Fibonacci numbers,   
\begin{eqnarray}
\sum_{k=1}^n h^{n-k} (h')_{1;{n-k}} ~F_{2 k}^{(h, h' + n - k)} & = & F_{2 n + 1}^{( h, h')} - h^{n} (h')_{1;{n}} 
\end{eqnarray} 
\end{theorem}  

\begin{proof}  
~\\
Using the h-Fibonacci recurrence relations, we have : 
\begin{eqnarray*}
F_{2 n}^{( h, h' )} & = & F_{2 n + 1}^{( h, h' )} - h h' ~F_{2 n - 1}^{(h, h' + 1 )} \\
F_{2 n - 2}^{(h, h'+1 )} & = & F_{2 n - 1}^{(h, h'+1 )} - h (h'+1) ~F_{2 n - 3}^{(h, h'+2 )} \\
F_{2 n - 4}^{(h, h'+2 )} & = & F_{2 n - 3}^{(h, h'+2)} - h (h'+2) ~F_{2 n - 5}^{(h, h'+3 )} \\
& \vdots & \\ 
F_{2 k}^{(h, h'+ n -k )} & = & F_{2 k + 1}^{(h, h'+n -k)} - h (h'+ n -k) ~F_{2 k - 1}^{(h, h'+n -k + 1 )} \\
& \vdots & \\
F_4^{(h,h' + n - 2 )} & = & F_5^{(h, h'+n-2 )} - h ( h' + n - 2) ~F_3^{(h, h'+n-1 )} \\
F_2^{(h,h'+n-1 )} & = & F_3^{(h, h'+n-1 )} - h ( h' +n -1 ) ~F_1^{(h,h'+n )} 
\end{eqnarray*}   
Multiplying $F_{2 k}^{(h,h'+ n -k)}$ by $h^{n-k} (h')_{1;{n-k}}$ and adding these equations, we get : 
\begin{eqnarray*}
\sum_{k=1}^n h^{n-k} (h')_{1;{n-k}} ~F_{ 2 k}^{(h,h' + n - k)} & = & F_{2 n + 1}^{(h,h' )} - h^{n} (h')_{1;{n}} ~F_1^{(h, h'+n)} \\
& = & F_{2 n + 1}^{( h,h' )} - h^{n} (h')_{1;{n}} 
\end{eqnarray*} 
\end{proof} 

Another way of introducing the Fibonacci numbers is to use the $Q$-matrix formulation where $Q$ is given by : 
\begin{eqnarray*} 
Q & = & \left( \begin{array}{cc} 
1 & 1 \\
1 & 0 \\
\end{array} \right) 
\end{eqnarray*} 
Now by raising $Q$ to the nth power, it can be shown that : 
\begin{eqnarray*}
Q^n & = & \left( \begin{array}{cc}
f_{n+1} & f_n \\
f_n & f_{n-1} \\ 
\end{array} \right) 
\end{eqnarray*}
where $n = \pm 1, \pm 2, \pm 3, \cdots$. \\

In the next section, we will introduce the h-Fibonacci matrices based on h-Fibonacci operators. Here the $Q_h$-matrix operators are utilized which 
are a generalization of the $Q$-matrix that depends on the parameter $h$. 
\section{Matrix Representation of h-Fibonacci Numbers} 
Since the h-Fibonacci numbers involve the shifted factorial, it is convenient for us to use repeated derivations to handle it : 
\begin{eqnarray}
\left( - h \frac{d}{dt} \right)^k t^{-h'} |_{t =1} & = & h^k \left( h' \right)_{1;k} 
\end{eqnarray} 
The latter equation permits us to define what we call here the h-Fibonacci operators as follows : 
\begin{eqnarray} 
{\bf F}_n  & = & \sum_{k=0}^{[\frac{n-1}{2}]} \left( \begin{array}{c} 
n - 1 - k \\
k \\ \end{array} \right) \left( -h \frac{d}{dt} \right)^k 
\end{eqnarray}   
where ${\bf F}_0 = 0, {\bf F}_1 = 1$. \\ 

With these operators, the h-Fibonacci numbers can be expressed as : 
\begin{eqnarray} 
F_n^{(h,h')} & = & {\bf F}_n t^{- h'}|_{t = 1} 
\end{eqnarray}
Moreover, it is easy to see that the Fibonacci operators obey the following recurrence formula : 
\begin{eqnarray} 
{\bf F}_{n+1} & = & {\bf F}_n  - h \frac{d}{dt} {\bf F}_{n-1} 
\label{eqoprec}
\end{eqnarray}
This sequence of operators can be extended to negative subscripts by defining them as : 
\begin{eqnarray}
-h \frac{d}{dt} {\bf F}_{-n}  & = & - {\bf F}_{-(n - 1)}  +  {\bf F}_{-(n - 2)} 
\end{eqnarray}

To reformulate the h-Fibonacci numbers in a matrix representation, let use first consider the $2 \times 2$ matrix operators 
${\bf Q_h}$,
\begin{eqnarray}
{\bf Q_h} & = & \left( \begin{array}{cc} 
1 & 1 \\
- h \frac{d}{dt} & 0 \\ 
\end{array} \right) 
\end{eqnarray} 
which can be represented in terms of the h-Fibonacci operators as follows,
\begin{eqnarray}
{\bf Q_h} & = &  \left( \begin{array}{cc} 
{\bf F}_2 & {\bf F}_1 \\
- h \frac{d}{dt} {\bf F}_1 &  -h \frac{d}{dt} {\bf F}_0 \\ 
\end{array} \right) 
\end{eqnarray} 
Then in general, for the nth power of the $Q_h$-matrix, we will get :
\begin{eqnarray}
{\bf Q_h}^n  & = & \left( \begin{array}{cc} 
{\bf F}_{n+1} & {\bf F}_n   \\
- h \frac{d}{dt} {\bf F}_n  & - h \frac{d}{dt} {\bf F}_{n-1}  \\ 
\end{array} \right) 
\label{eqAn}
\end{eqnarray} 

\begin{theorem} 
~\\
For any given $n > 0$, the following property holds for the nth power of the $Q_{h,h'}$-matrix :
\begin{eqnarray}
Q_{h,h'}^n  & = & \left( \begin{array}{cc} 
F_{n+1}^{(h,h') }  & F_n^{( h,h' )}  \\
h h' ~F_n^{(h,h'+1 )}  & h h' ~F_{n-1}^{(h,h'+1 )} \\ 
\end{array} \right)  
\label{qhn}
\end{eqnarray}
where $Q_{h,h'}^n = {\bf Q_h}^n t^{- h'} \left|_{t=1} \right.$ 
\end{theorem}   
\begin{proof}
~\\
The proof of this theorem is straightforward using definitions. 
\end{proof} 

By taking the inverse of the ${\bf Q_h}$-matrix, it is easy to find that :  
\begin{eqnarray}
( -h \frac{d}{dt} ) {\bf Q_h^{-1} } & = & - \left( \begin{array}{cc} 
- h \frac{d}{dt} {\bf F}_0 & - {\bf F}_1 \\
- ( - h \frac{d}{dt} ) {\bf F}_1 & {\bf F}_2 \\ \end{array} \right) 
\end{eqnarray}
and in general, the inverse of the ${\bf Q_h}$-matrix to the nth power can be written as :
\begin{eqnarray}
 (- h \frac{d}{dt})^n ~{\bf Q_h^{-n} }  & = &  (- 1 )^n \left( \begin{array}{cc} 
- h \frac{d}{dt} {\bf F}_{n-1} & - {\bf F}_n \\
- (-h \frac{d}{dt}) {\bf F}_n & {\bf F}_{n+1} \\ \end{array} \right) 
\label{qh-n}
\end{eqnarray} 
which means that : 
\begin{eqnarray}
(- h \frac{d}{dt})^n ~{\bf F}_{-n}   & = & 
( - 1 )^{n-1} ~{\bf F}_n  
\end{eqnarray}
This result implies that the h-Fibonacci numbers with negative indices can be expressed in terms of the positive indices.  

\begin{theorem} 
~\\
For $n \ge 0$, we have the following property that relates the h-Fibonacci number with negative index to the one with a positive index, 
\begin{eqnarray} 
\left( h h' \right)^n  ~F_{-n}^{(h,h' + n)} & = & ( - 1 )^{n-1} ~F_n^{(h,h' )} 
\label{oppone}
\end{eqnarray}
\end{theorem} 
\begin{proof} 
~\\ 
By definition, we have 
\begin{eqnarray*} 
\left(- h \frac{d}{dt} \right)^n {\bf F}_{-n} t^{-h'}|_{t=1} & = & 
\left( h h' \right)^n ~F_{-n}^{(h,h' + n )}
\end{eqnarray*} 
Now using equation ( \ref{oppone} ) we get,
\begin{eqnarray*} 
\left( h h' \right)^n ~F_{-n}^{(h,h' + n)} & = & (- 1)^n {\bf F}_{n} t^{-h'}|_{t=1} \nonumber \\
& = & ( - 1 )^{n-1} ~F_n^{(h,h' )} 
\end{eqnarray*} 
\end{proof}  
Next we derive some nice identities between h-Fibonacci operators. 

\begin{theorem}
~\\
Let $n$ be a positive integer. Then 
\begin{eqnarray}
{\bf F}_{n+1} ~{\bf F}_{n-1} - {\bf F}_n ^2 & = & ( - 1 )^n  \left(- h \frac{d}{dt} \right)^{n-1} 
\end{eqnarray} 
\end{theorem}
\begin{proof}
~\\
This theorem is easily proven by taking the determinant in equation ( \ref{qhn} ) and using the fact that 
$det( {\bf Q_h}^n ) = ( det {\bf Q_h} )^n$. 
\end{proof} 

\begin{theorem}
~\\
Let $n$ and $m$ be positive integers. Then we have, 
\begin{eqnarray}  
{\bf F}_{m+n+1}  & = & {\bf F}_{m+1} {\bf F}_{n+1}  - h \frac{d}{dt} {\bf F}_m {\bf F}_n  \nonumber \\ 
{\bf F}_{m+n}    & = & {\bf F}_{m+1} {\bf F}_n  - h \frac{d}{dt} {\bf F}_m {\bf F}_{n-1}   \nonumber \\   
{\bf F}_{m+n}    & = & {\bf F}_m {\bf F}_{n+1}  - h \frac{d}{dt} {\bf F}_{m-1} {\bf F}_n  \nonumber \\  
{\bf F}_{m+n-1}  & = & {\bf F}_m {\bf F}_n  - h \frac{d}{dt} {\bf F}_{m-1} {\bf F}_{n-1}   
\end{eqnarray}  
\end{theorem}
\begin{proof}
~\\
The proof of this theorem is straightforward by using that ${\bf Q_h}^{m+n} = {\bf Q_h}^m {\bf Q_h}^n$ 
and equating the corresponding matrix entries.  
\end{proof}

Since $Q_h$-matrix is a $2 \times 2$ matrix, the matrix powers of this matrix are not independent. Indeed the Cayley-Hamilton theorem implies that : 
\begin{eqnarray}
{\bf Q}_h^2  & = & {\bf Q}_h - h \frac{d}{dt} {\bf I} 
\end{eqnarray}
Hence : 
\begin{eqnarray}
{\bf Q}_h^k  & = & {\bf F}_k {\bf Q}_h  - h \frac{d}{dt} {\bf F}_{k-1} {\bf I} 
\label{qhk}
\end{eqnarray}

\begin{theorem} 
~\\
We have 
\begin{eqnarray}
 \sum_{i=0}^n \left( \begin{array}{cc} 
 n \\ 
 i \\ \end{array} \right) 
 \left( -h \frac{d}{dt} \right)^{n-i} {\bf F}_k^i {\bf F}_{k-1}^{n-i} {\bf F}_i & = & {\bf F}_{k n} 
 \end{eqnarray}
 \end{theorem} 
 
 \begin{proof}
 ~\\
 By taking the nth power of equation ( \ref{qhk} ), we get 
 \begin{eqnarray*}
 {\bf Q}_h^{kn} = \left( {\bf F}_k {\bf Q}_h - h \frac{d}{dt} {\bf F}_{k-1} {\bf I} \right)^n \\
 & = & \sum_{i=0}^n \left( \begin{array}{cc} 
 n \\ 
 i \\ \end{array} \right) {\bf F}_k^i \left( - h \frac{d}{dt} \right)^{n-i} {\bf F}_{k-1}^{n-i} {\bf Q}_h^i \\
 & = & \sum_{i=0}^n \left( \begin{array}{cc} 
 n \\ 
 i \\ \end{array} \right) {\bf F}_k^i \left( - h \frac{d}{dt} \right)^{n-i} {\bf F}_{k-1}^{n-i} 
 \left( {\bf F_i} {\bf Q}_h - h \frac{d}{dt} {\bf I} \right) 
 \end{eqnarray*}
 \end{proof}  
 By equating the 21 matrix element above, we get : 
 \begin{eqnarray*}
 {\bf F}_{k n} & = & \sum_{i=0}^n \left( \begin{array}{cc} 
 n \\ 
 i \\ \end{array} \right) 
 \left( -h \frac{d}{dt} \right)^{n-i} {\bf F}_k^i {\bf F}_{k-1}^{n-i} {\bf F}_i
 \end{eqnarray*} 
 Similarly, using the Cayley-Hamilton theorem we can re-express equation ( \ref{qh-n} ) as : 
 \begin{eqnarray}
 \left(- h \frac{d}{dt} \right)^k {\bf Q}_h^{-k} & = & \left(- 1 \right)^{k+1} \left( {\bf F}_k {\bf Q}_h - {\bf F}_{k+1} {\bf I} \right) 
 \label{qh-k}
 \end{eqnarray} 
 
 \begin{theorem} 
 ~\\
 We have 
 
 \begin{eqnarray} 
 \sum_{i=0}^n \left( \begin{array}{cc} 
 n \\ 
 i \\ \end{array} \right) \left(- 1 \right)^{i+1} {\bf F}_k^i {\bf F}_{k+1}^{n-i} {\bf F}_i  & = & {\bf F}_{k n} 
 \end{eqnarray}
 \end{theorem} 
 
 \begin{proof} 
 ~\\ 
 By taking the nth power of equation ( \ref{qh-k} ), we get 
 \begin{eqnarray*} 
 \left(- h \frac{d}{dt} \right)^k {\bf Q}_h^{-k} & = & \left(- 1 \right)^{ k n + n} \left( {\bf F}_k {\bf Q}_h - {\bf F}_{k+1} {\bf I} \right) \\
 & = & \left(- 1 \right)^{ k n + n} \sum_{i=0}^n \left( \begin{array}{cc} 
 n \\ 
 i \\ \end{array} \right) {\bf F}_k^i \left(- 1 \right)^{n-i} {\bf F}_{k+1}^{n-i} {\bf Q}_h^i \\
 & = & \left(- 1 \right)^{k n} \sum_{i=0}^n \left( \begin{array}{cc} 
 n \\ 
 i \\ \end{array} \right) \left(- 1 \right)^i {\bf F}_{k+1}^{n-i} \left( {\bf F}_i {\bf Q}_h - h \frac{d}{dt} {\bf F}_i {\bf I} \right)
 \end{eqnarray*} 
 The 21 matrix element of the above matrix gives : 
 \begin{eqnarray*} 
 \left(- 1 \right)^{ k n + 1} {\bf F}_{k n} & = & \left(- 1 \right)^{k n} \sum_{i=0}^n \left( \begin{array}{cc} 
 n \\ 
 i \\ \end{array} \right) \left(- 1 \right)^i {\bf F}_k^i {\bf F}_{k+1}^{n-i} {\bf F}_i 
 \end{eqnarray*} 
 which means 
 \begin{eqnarray*}
 {\bf F}_{k n} & = & \sum_{i=0}^n \left( \begin{array}{cc} 
 n \\ 
 i \\ \end{array} \right) \left(- 1 \right)^{i+1} {\bf F}_k^i {\bf F}_{k+1}^{n-i} {\bf F}_i 
 \end{eqnarray*}
 \end{proof}
\section{h-Binet Formula } 
Binet's formula is well known in the Fibonacci numbers theory. In this section, we derive the h-Binet's formula for the h-Fibonacci numbers using the $Q_h$-matrix formulation. 

\begin{theorem}  
~\\
for all $n \geq 0$, we have for the h-Fibonacci operators 
\begin{eqnarray}
{\bf F}_n  & = &  \frac{1}{\lambda_+ - \lambda_-}~ 
\left( \lambda_+^n - \lambda_-^n \right)
\end{eqnarray} 
where $\lambda_{\pm}$ are the eigenvalues of the matrix ${\bf Q_h} (1, - h \frac{d}{dt})$ given by : 
\begin{eqnarray}
\lambda_{\pm} & = & \frac{1 \pm \sqrt{1 + 4 (-h \frac{d}{dt})}}{2} 
\label{zeroesofA}
\end{eqnarray}
\end{theorem} 

\begin{proof} 
~\\ 
The eigenvalues of the matrix ${\bf Q_h}$ are obtained from : 
\begin{eqnarray*} 
\left| \begin{array}{cc} 
1 - \lambda & 1 \\
-h \frac{d}{dt} & - \lambda \end{array} \right| & = & 0  
\end{eqnarray*} 
By solving the determinant for $\lambda$, we get the two real solutions $\lambda_{\pm}$ 
in equation ( \ref{zeroesofA}). \\

Now the matrix ${\bf Q_h}$ can be written in terms of the eigenvalues $\lambda_{\pm}$ and eigenvectors as : 
\begin{eqnarray*}
{\bf Q_h} & = & \left( \begin{array}{cc} 
1 & 1 \\
- \lambda_- & - \lambda_+ \end{array} \right) 
\left( \begin{array}{cc} 
\lambda_+ & 0 \\
0 & \lambda_- \end{array} \right) 
\left( \begin{array}{cc} 
1 & 1 \\
- \lambda_- & - \lambda_+ \end{array} \right)^{-1} 
\end{eqnarray*} 
From this, we obtain the ${\bf Q_h}$ matrix to the nth power,
\begin{eqnarray*}
{\bf Q_h}^n  & = & \frac{1}{\lambda_- - \lambda_+}~
\left( \begin{array}{cc} 
1 & 1 \\
- \lambda_- & - \lambda_+ \end{array} \right) 
\left( \begin{array}{cc} 
\lambda_+ ^n& 0 \\
0 & \lambda_-^n \end{array} \right) 
\left( \begin{array}{cc} 
- \lambda_+ & -1 \\
\lambda_- & 1 \end{array} \right) \nonumber \\
& = & \frac{1}{\lambda_+ - \lambda_-}~ \left( \begin{array}{cc} 
\lambda_+^{n+1} - \lambda_-^{n+1} & \lambda_+^n - \lambda_-^n \\
- \lambda_+ \lambda_- ( \lambda_+^n - \lambda_-^n ) & 
- \lambda_+ \lambda_- ( \lambda_+^{n-1} - \lambda_-^{n-1} ) \\ \end{array} \right) 
\end{eqnarray*}  
We finally get : 
\begin{eqnarray*}
{\bf F}_n & = & \frac{1}{\lambda_+ - \lambda_-}~ 
\left( \lambda_+^n - \lambda_-^n \right) \nonumber \\
&= & \frac{\left( \frac{1 + \sqrt{1 + 4 (-h \frac{d}{dt})}}{2} \right)^n - 
\left( \frac{1 - \sqrt{1 + 4 (-h \frac{d}{dt})}}{2} \right)^n }{\sqrt{1 + 4 (-h \frac{d}{dt})}} 
\end{eqnarray*}
\end{proof} 
  
We use the Binet's formula to derive some identities between h-Fibonacci operators and numbers. 

\begin{theorem}
( h-Catalan's identity ) \\

The following property holds for h-Fibonacci operators :  
\begin{eqnarray}
{\bf F}_{n-m} {\bf F}_{n+m}  - {\bf F}_n^2  & = & \left( - 1 \right)^{n+1-m} \left( - h \frac{d}{dt} \right)^{n-m} {\bf F}_m^2 \nonumber \\ 
\end{eqnarray}
\end{theorem}

\begin{proof}
~\\
Using h-Binet's formula, we have : 
\begin{eqnarray*}
{\bf F}_{n-m} {\bf F}_{n+m}  - {\bf F}_n^2 & = & \frac{\lambda_+^{n-m} - \lambda_-^{n-m}}{\lambda_+ - \lambda_-} . 
\frac{\lambda_+^{n+m} - \lambda_-^{n+m}}{\lambda_+ - \lambda_-} - 
\left( \frac{\lambda_+^n - \lambda_-^n}{\lambda_+ - \lambda_-} \right)^2 \nonumber \\
& = & \frac{- \lambda_+^{n-m} \lambda_-^{n+m} - \lambda_+^{n+m} \lambda_-{n-m} + 2 \lambda_+^n \lambda_-^n}{\left(\lambda_+ - \lambda_- \right)^2} \\ 
& = & - \frac{\left( \lambda_+ \lambda_- \right)^n}{\left(\lambda_+ - \lambda_- \right)^2} \left( (\frac{\lambda_-}{\lambda_+})^m + 
(\frac{\lambda_+}{\lambda_-})^m - 2 \right) \\ 
& = & - \frac{\left( \lambda_+ \lambda_- \right)^{n-m}}{\left(\lambda_+ - \lambda_- \right)^2} \left( \lambda_+^{2 m} + 
\lambda_-^{2 m} - 2 \lambda_+^m \lambda_-^m \right) \\ 
& = & - \left( \lambda_+ \lambda_- \right)^{n-m} \left( \frac{\lambda_+^m - \lambda_-^m}{\lambda_+ - \lambda_-} \right)^2 \\ 
& = & - \left( \lambda_+ \lambda_- \right)^{n-m} {\bf F}_m^2 
\end{eqnarray*}  
Now since $\lambda_+ \lambda_- = h \frac{d}{dt}$, we get 
\begin{eqnarray*}
{\bf F}_{n-m} {\bf F}_{n+m} - {\bf F}_n^2 & = & \left( - 1 \right)^{n+1-m} \left( - h \frac{d}{dt} \right)^{n-m} {\bf F}_m^2 
\end{eqnarray*}
\end{proof} 

\begin{theorem} 
( h- d'Ocagne identity ) \\
If $n > m$, then 
\begin{eqnarray}
{\bf F}_m {\bf F}_{n+1} - {\bf F}_{m+1} {\bf F}_n & = & \left(-1 \right)^n \left(- h \frac{d}{dt} \right)^n {\bf F}_{m-n}  \nonumber \\
\end{eqnarray}
\end{theorem}

\begin{proof} 
~\\ 
Using again the h-Binet's formula, we have : 
\begin{eqnarray*}
{\bf F}_m {\bf F}_{n+1} - {\bf F}_{m+1} {\bf F}_n & = & \frac{\lambda_+^m - \lambda_-^m}{\lambda_+ - \lambda_-} . 
\frac{\lambda_+^{n+1} - \lambda_-^{n+1}}{\lambda_+ - \lambda_-} - 
\frac{\lambda_+^{m+1} - \lambda_-^{m+1}}{\lambda_+ - \lambda_-} . 
\frac{\lambda_+^n - \lambda_-^n}{\lambda_+ - \lambda_-} \\ 
& = & \frac{- \lambda_+^{n+1} \lambda_-^m - \lambda_+^m \lambda_-^{n+1} + \lambda_+^n \lambda_-^{m+1} + \lambda_+^{m+1} \lambda_-^n}{\left( 
\lambda_+ - \lambda_- \right)^2} \\
& = & \frac{- \lambda_+^n \lambda_-^m ( \lambda_+ - \lambda_- ) + \lambda_+^m \lambda_-^n ( -\lambda_+ + \lambda_- )}{\left( 
\lambda_+ - \lambda_- \right)^2} \\ 
& = & \frac{- \lambda_+^n \lambda_-^m + \lambda_+^m \lambda_-^n}{\lambda_+ - \lambda_-} \\
& = & \left( \lambda_+ \lambda_- \right)^n \frac{\lambda_+^{m-n} - \lambda_-^{m-n}}{\lambda_+ - \lambda_-} \\  
& = & \left( - 1 \right)^n \left( - h \frac{d}{dt} \right)^n {\bf F}_{m-n}
\end{eqnarray*}
\end{proof}

\begin{theorem} 
~\\ 
\begin{eqnarray}
\sum_{i=1}^n \left( \begin{array}{c} 
n \\ 
i \\ \end{array} \right) 
\left( - h \frac{d}{dt} \right)^{n-i} {\bf F}_i & = & {\bf F}_{2 n} 
\end{eqnarray}
\end{theorem} 

\begin{proof} 
~\\
By using the h-Binet's formula, we have : 
\begin{eqnarray*} 
\sum_{i=1}^n \left( \begin{array}{c} 
n \\ 
i \\ \end{array} \right) 
\left( - h \frac{d}{dt} \right)^{n-i} {\bf F}_i & = & \sum_{i=1}^n \left( \begin{array}{c} 
n \\ 
i \\ \end{array} \right) 
\left( - \lambda_+ \lambda_- \right)^{n-i} \frac{\lambda_+^i - \lambda_-^i}{\lambda_+ - \lambda_-} \\
& = & \frac{1}{\lambda_+ - \lambda_-} \left( \sum_{i=1}^n \left( \begin{array}{c} 
n \\ 
i \\ \end{array} \right) 
\left( - \lambda_+ \lambda_- \right)^{n-i} \lambda_+^i - \sum_{i=1}^n \left( \begin{array}{c} 
n \\ 
i \\ \end{array} \right) 
\left( - \lambda_+ \lambda_- \right)^{n-i} \lambda_-^i \right) \\
& = & \frac{1}{\lambda_+ - \lambda_-} \left( ( - \lambda_+ \lambda_- + \lambda_+ )^n - ( - \lambda_+ \lambda_- + \lambda_- )^n \right) \\
& = & \frac{1}{\lambda_+ - \lambda_-} \left( \lambda_+^n ( - \lambda^- + 1 )^n - \lambda_-^n ( - \lambda_+ + 1 )^n \right) \\
& = & \frac{\lambda_+^{2 n} - \lambda_-^{2 n}}{\lambda_+ - \lambda_-} \\
& = & {\bf F}_{2 n} 
\end{eqnarray*}  
where we have used $\lambda_{\pm} = 1 - \lambda_{\mp}$.
\end{proof}  

\begin{theorem}
~\\ 
\begin{eqnarray}
\sum_{i=1}^n \left( \begin{array}{c} 
n \\ 
i \\ \end{array} \right) F_i^{(h,h' + n - i)} & = & F_{2 n}^{(h,h')} 
\end{eqnarray}
\end{theorem} 

\begin{proof} 
~\\
This theorem is easily proven using the latter theorem. 
\end{proof}

\begin{theorem}
~\\ 
\begin{eqnarray} 
\sum_{i=1}^n \left( \begin{array}{c} 
n \\ 
i \\ \end{array} \right) \left(- 1 \right)^{n-i} {\bf F}_i & = & \left(- 1 \right)^{n-1} {\bf F}_n 
\end{eqnarray}
\end{theorem} 

\begin{proof} 
~\\ 
\begin{eqnarray*} 
\sum_{i=1}^n \left( \begin{array}{c} 
n \\ 
i \\ \end{array} \right) \left(- 1 \right)^{n-i} {\bf F}_i & = &  \sum_{i=1}^n \left( \begin{array}{c} 
n \\ 
i \\ \end{array} \right) \left(- 1 \right)^{n-i} \frac{\lambda_+^i - \lambda_-^i}{\lambda_+ - \lambda_-} \\
& = & \frac{1}{\lambda_+ - \lambda_-} \left( \sum_{i=1}^n \left( \begin{array}{c} 
n \\ 
i \\ \end{array} \right) \left(- 1 \right)^{n-i} \lambda_+^i - \sum_{i=1}^n \left( \begin{array}{c} 
n \\ 
i \\ \end{array} \right) \left(- 1 \right)^{n-i} \lambda_-^i \right) \\
& = & \frac{(-1 + \lambda_+)^n - (-1 \lambda_-)^n}{\lambda_+ - \lambda_-} \\
& = & \frac{(- \lambda_-)^n - (- \lambda_+)^n}{\lambda_+ - \lambda_-} \\
& = & \left(- 1 \right)^{n-1} \frac{\lambda_+^n - \lambda_-^n}{\lambda_+ - \lambda_-} \\
& = & \left(- 1 \right)^{n-1} {\bf F}_n 
\end{eqnarray*}
\end{proof} 

\begin{theorem} 
~\\
\begin{eqnarray} 
\sum_{i=1}^n \left( \begin{array}{c} 
n \\ 
i \\ \end{array} \right) \left(- 1 \right)^{n-i} F_i^{(h,h')} & = & \left(- 1 \right)^{n-1} F_n^{(h,h')} 
\end{eqnarray}
\end{theorem} 

\begin{proof}
~\\
The proof follows from the latter theorem. 
\end{proof} 

\section{Generating Function for h-Fibonacci Operators}  
In this section, the generating functions for the h-Fibonacci operators are given. As a result, the h-Fibonacci operator sequences are seen as the coefficients of the power series of the corresponding generating function. 

To derive a generating function for h-Fibonacci operators, consider the function ${\bf g} ( x )$ given by : 
\begin{eqnarray}
{\bf g} ( x ) & = & \sum_{k=0}^{\infty} {\bf F}_k ~x^k 
\end{eqnarray}
It follows that 
\begin{eqnarray*}
{\bf g} ( x ) - {\bf F}_0 ~x^0 - {\bf F}_1 ~x^1 & = & {\bf g} ( x ) - x   
\end{eqnarray*}
Hence 
\begin{eqnarray*}
{\bf g} ( x ) - x & = & \sum_{k=2}^{\infty} {\bf F}_k ~x^k  ~=~ 
\sum_{k=2}^{\infty} \left( {\bf F}_{k-1} ~x^k + (-h \frac{d}{dt}) {\bf F}_{k-2} ~x^k  \right) \nonumber \\
& = & x \sum_{k=1}^{\infty} {\bf F}_k ~x^k  - \lambda_+ \lambda_- x^2 \sum_{k=0}^{\infty} {\bf F}_k ~x^k  \nonumber \\
& = & x ~{\bf g} ( x ) - \lambda_+  \lambda_- ~x^2 ~{\bf g} ( x ) 
\end{eqnarray*} 
From which we get : 
\begin{eqnarray*}
\left( 1 - x + \lambda_+ \lambda_- ~x^2 \right) ~{\bf g} ( x ) & = & x  
\end{eqnarray*}
Thus  
\begin{eqnarray*} 
{\bf g} ( x ) & = & \frac{x}{1 - x + \lambda_+ \lambda_- ~x^2} ~=~ \frac{x}{(1 - \lambda_+ x ) (1 - \lambda_- x )} 
\end{eqnarray*}
where we have used $\lambda_+ + \lambda_- = 1$. \\

So the generating function for the h-Fibonacci operators is : 
\begin{eqnarray}
{\bf g} ( x ) & = &  \frac{x}{(1 - \lambda_+ x ) (1 - \lambda_- x )} 
~=~ \sum_{k=0}^{\infty} {\bf F}_k ~x^k  
\end{eqnarray}  

Next we list the generating functions that generate the various powers and products of the h-Fibonacci sequences. 

\begin{theorem}
~\\
We have 
\begin{eqnarray}
\frac{x}{\left(1 - \lambda_+^2 x \right) \left(1 - \lambda_-^2 x \right)} 
& = & \sum_{k=0}^{\infty} {\bf F}_{2 k} ~x^k 
\end{eqnarray}
\end{theorem} 

\begin{proof}
~\\
Using the h-Binet formula, we have : 
\begin{eqnarray*} 
\sum_{k=0}^{\infty} {\bf F}_{2 k} x^k & = & \sum_{k=0}^{\infty} \frac{\lambda_+^{2 k} - \lambda_-^{2 k}}{\lambda_+ - \lambda_-} x^k \\
& = & \frac{1}{\lambda_+ - \lambda_-} \left( \sum_{k=0}^{\infty} \lambda_+^{2 k} x^k - \sum_{k=0}^{\infty} \lambda_-^{2 k} x^k \right) \\
& = & \frac{1}{\lambda_+ - \lambda_-} \left( \frac{1}{1 - \lambda_+^2 x} - \frac{1}{1 - \lambda_-^2 x} \right) \\
& = & \frac{( \lambda_+^2 - \lambda_-^2) x}{( \lambda_+ - \lambda_- ) ( 1 - \lambda_+^2 x ) (1 - \lambda_-^2 x)} \\
& = & \frac{x}{ (1 - \lambda_+^2 x ) (1 - \lambda_-^2 x)}
\end{eqnarray*} 
\end{proof} 

\begin{theorem} 
~\\ 
We have
\begin{eqnarray}
\frac{1 + \lambda_+ \lambda_- x}{ \left(1 - \lambda_+^2 x \right) \left(1 - \lambda_-^2 x \right)} 
& = & \sum_{k=0}^{\infty} {\bf F}_{2 k + 1} ~x^k 
\end{eqnarray}
\end{theorem} 

\begin{proof}
~\\
Using the h-Binet formula, we have : 
\begin{eqnarray*} 
\sum_{k=0}^{\infty} {\bf F}_{2 k + 1} x^k & = & \sum_{k=0}^{\infty} \frac{\lambda_+^{2 k+ 1} - \lambda_-^{2 k+ 1}}{\lambda_+ - \lambda_-} x^k \\
& = & \frac{1}{\lambda_+ - \lambda_-} \left( \sum_{k=0}^{\infty} \lambda_+^{2 k+ 1} x^k - \sum_{k=0}^{\infty} \lambda_-^{2 k+1} x^k \right) \\
& = & \frac{1}{\lambda_+ - \lambda_-} \left( \frac{\lambda_+}{1 - \lambda_+^2 x} - \frac{\lambda_-}{1 - \lambda_-^2 x} \right) \\
& = & \frac{\lambda_+ - \lambda_- + \lambda_+ \lambda_- (\lambda_+ - \lambda_- ) x }
{( \lambda_+ - \lambda_- ) ( 1 - \lambda_+^2 x ) (1 - \lambda_-^2 x)} \\
& = & \frac{1 + \lambda_+ \lambda_- x}{ (1 - \lambda_+^2 x ) (1 - \lambda_-^2 x)}
\end{eqnarray*} 
\end{proof} 

\begin{theorem} 
~\\ 
We have
\begin{eqnarray}
\frac{{\bf F}_m + \lambda_+ \lambda_- {\bf F}_{m-1} x}{\left(1 - \lambda_+ x \right) \left(1 - \lambda_- x \right)} 
& = & \sum_{k=0}^{\infty} {\bf F}_{m+n} ~x^n 
\end{eqnarray}
\end{theorem}   

\begin{proof}
~\\
Using the h-Binet formula, we have : 
\begin{eqnarray*} 
\sum_{n=0}^{\infty} {\bf F}_{m+n} x^n & = & \sum_{n=0}^{\infty} \frac{\lambda_+^{m+n} - \lambda_-^{m+n}}{\lambda_+ - \lambda_-} x^n \\
& = & \frac{1}{\lambda_+ - \lambda_-} \left( \sum_{n=0}^{\infty} \lambda_+^{m+n} x^n - \sum_{n=0}^{\infty} \lambda_-^{m+n} x^n \right) \\
& = & \frac{1}{\lambda_+ - \lambda_-} \left( \frac{\lambda_+^m}{1 - \lambda_+ x} - \frac{\lambda_-^m}{1 - \lambda_- x} \right) \\
& = & \frac{\lambda_+^m - \lambda_-^m + \lambda_+ \lambda_- (\lambda_+^{m-1} - \lambda_-^{m-1} ) x }
{( \lambda_+ - \lambda_- ) ( 1 - \lambda_+^2 x ) (1 - \lambda_-^2 x)} \\
& = & \frac{ {\bf F}_m - \lambda_+ \lambda_- {\bf F}_{m-1} x }{ (1 - \lambda_+  x ) (1 - \lambda_- x)}
\end{eqnarray*} 
\end{proof} 

\begin{theorem}
~\\
We have 
\begin{eqnarray} 
\frac{x + \lambda_+ \lambda_- x^2}{\left(1 - \lambda_+^2 x \right) \left(1 - \lambda_-^2 x \right) \left(1 - \lambda_+ \lambda_- x \right)} 
& = & \sum_{k=0}^{\infty} {\bf F}_k^2 ~x^k \\
\frac{x}{\left(1 - \lambda_+^2 x \right) \left(1 - \lambda_-^2 x \right) \left(1 - \lambda_+ \lambda_- x \right)} 
& = & \sum_{k=0}^{\infty} {\bf F}_k {\bf F}_{k+1} ~x^k \\ 
\frac{1}{\left(1 - \lambda_+^2 x \right) \left(1 - \lambda_-^2 x \right) \left(1 - \lambda_+ \lambda_- x \right)} 
& = & \sum_{k=0}^{\infty} {\bf F}_{k+1} {\bf F}_{k+2} ~x^k  \\
\frac{x + 2 \lambda_+ \lambda_- x^2 + \lambda_+^3 \lambda_-^3 x^3}
{\left(1 - \lambda_+^3 x \right) \left(1 - \lambda_-^3 x \right) \left(1 - \lambda_+^2 \lambda_- x \right) \left(1 - \lambda_+ \lambda_-^2 x \right)} 
& = & \sum_{k=0}^{\infty} {\bf F}_k^3 ~x^k 
\end{eqnarray}
\end{theorem} 

\begin{proof}
~\\
Easy to prove using h-Binet's formula. 
\end{proof}

The following proposition gives us the value for the h-Fibonacci sequence series with weights $p^{-(i+1)}$. 

\begin{theorem} 
~\\
For each non-vanishing integer number p :   
\begin{eqnarray}
\sum_{i=0}^{\infty} \frac{{\bf F}_i}{p^{i+1}} & = & \frac{1}{p^2 - p + \lambda_+ \lambda_-}  
\label{hseriesw}
\end{eqnarray} 
\end{theorem}

\begin{proof} 
~\\
Using the h-Binet's formula, we get : 
\begin{eqnarray*}
\sum_{i=0}^{\infty} \frac{{\bf F}_i}{p^{i+1}} & = & \sum_{i=0}^{\infty} \frac{1}{p^{i+1}} . 
\frac{\lambda_+^i - \lambda_-^i}{\lambda_+ - \lambda_-} \\
& = & \frac{1}{p \left(\lambda_+ - \lambda_- \right)} \left( \sum_{i=0}^{\infty} 
\left( \frac{\lambda_+}{p} \right)^i - \sum_{i=0}^{\infty} 
\left( \frac{\lambda_-}{p} \right)^i \right) \\
& = & \frac{1}{\left( p - \lambda_+ \right) \left( p - \lambda_- \right)}  
\end{eqnarray*}
\end{proof} 

Equation ( \ref{hseriesw} ) yields the following results for particular values of $p$ : 
\begin{center}  
\begin{itemize}
\item When $p = 2$~~~~~~
$\sum_{i=0}^{\infty} \frac{{\bf F}_i}{2^{i+1}} = \frac{1}{1 + 1 + \lambda_+ \lambda_-} = 
\frac{1}{f_1 + 1 + \lambda_+ \lambda_-}$ \\
\item When $p = 3$~~~~~~
$\sum_{i=0}^{\infty} \frac{{\bf F}_i}{3^{i+1}} = \frac{1}{5 + 1 + \lambda_+ \lambda_-} = 
\frac{1}{f_5 + 1 + \lambda_+ \lambda_-}$ \\ 
\item When $p = 8$~~~~~~
$\sum_{i=0}^{\infty} \frac{{\bf F}_i}{8^{i+1}} = \frac{1}{55 + 1 + \lambda_+ \lambda_-} = 
\frac{1}{f_{10}+ 1 + \lambda_+ \lambda_-}$ \\ 
\item When $p = 10$~~~~
$\sum_{i=0}^{\infty} \frac{{\bf F}_i}{10^{i+1}} = \frac{1}{89 + 1 + \lambda_+ \lambda_-} = 
\frac{1}{f_{11} + 1 + \lambda_+ \lambda_-}$ \\
\end{itemize}  
\end{center} 
where $f_n$ appearing in the denominator is the usual Fibonacci number.
\section{Conclusions}
h-analogue of Fibonacci numbers have been introduced and studied. Several properties of these numbers are derived. In addition, 
the h-Binet's formula for these numbers is found and the generating function of these h-Fibonacci sequences and their various 
powers have been deduced. It is straightforward to introduce the h-Lucas numbers. This work is in progress \cite{benaoum2}.   

It is possible to introduce the $(q-h)$-analogue of Finonacci numbers by using the $(q-h)$-analogue of binomial coefficients which was introduced in 
\cite{benaoum3}. Indeed the $(q-h)$-analogue of binomial coefficients was found to be : 
\begin{eqnarray}
\left[ \begin{array}{c} 
n \\
k \\ \end{array} \right]_{(q,h,h')} & = &  \left[ \begin{array}{c} 
n \\
k \\ \end{array} \right]_q h^k \left(h' \right)_{1;[k]} = \left[ \begin{array}{c} 
n \\
k \\ \end{array} \right]_q  \left( h h' \right)_{h;[k]}
\end{eqnarray} 
These coefficients obey to the following properties :  

\begin{eqnarray}
\left[ 
\begin{array}{l} 
n + 1 \\
k \end{array} \right]_{(q,h,h')} & = &  q^k 
~\left[ 
\begin{array}{l} 
n  \\
k \end{array} \right]_{(q,h,h')} + h h' \left[ 
\begin{array}{l} 
n  \\
k - 1  \end{array} \right]_{(q,h,h'+ 1)} 
\label{equation18}
\end{eqnarray}
and 
\begin{eqnarray}
\left[ 
\begin{array}{l} 
n + 1 \\
k + 1 \end{array} \right]_{(q,h,h')} & = &  h h'~\frac{[n+1]_q}{[k+1]_q}~
\left[ 
\begin{array}{l} 
n  \\
k \end{array} \right]_{(q,h,h'+1)} 
\end{eqnarray}

The $(q-h)$-analogue of Fibonacci numbers will be defined as follows : 
\begin{eqnarray} 
{^q} F_{n+1}^{( h,h' )} & = & \sum_{k=0}^{[\frac{n-1}{2}]}
q^{k^2} 
\left[ \begin{array}{c} 
n - k \\
k \\ \end{array} \right]_{(q,h,h')}
\end{eqnarray} 
and they obey the following recurrence formula : 
\begin{eqnarray}
{^q} F_{n+1}^{( h,h')} & = & {^q} F_n^{( h,h' )} + q^{n-1}~ {^q} F_{n-1}^{(h,h' + 1)} 
\end{eqnarray} 
For $h h' = 1$ and $h = 0$, the $(q-h)$-analogue of Fibonacci numbers are just the q-Fibonacci numbers ( see \cite{cigler} ). \\
Similarly, several properties of the $(q-h)$-analogue of Fibonacci numbers can be derived. We list below some of them which are easy to prove from the 
recurrence formula. 
\begin{eqnarray}
h h' ~\sum_{k=0}^{n} q^k~  {^q} F_k^{(h,h'+1)} & = & {^q} F_{n+2}^{( h,h' )} - 1 \nonumber \\
\sum_{k=0}^{n} q^{2 k}~ h^{n-k} (h')_{1;n-k} ~{^q} F_{2 k - 1}^{(h,h' +n - k )} & = & {^q} F_{2 n}^{(h,h')}  \nonumber \\  
\sum_{k=0}^{n} h^{n-k} (h')_{1;n-k} ~{^q} F_{2 k}^{(h,h'+n-k)}/ q^{2 k - 1} & = & {^q} F_{2 n + 1}^{( h,h' )} -  h^{n} (h')_{1;n}
\end{eqnarray} 

\section*{Acknowledgments}
I would like to thank Tom Koornwinder for his helpful comments and suggestions.

\end{document}